\documentclass[runningheads, envcountsame, a4paper]{llncs}

\usepackage[utf8]{inputenc}

\usepackage{microtype}

\usepackage{amsfonts}
\usepackage{amsmath}

\usepackage{color}

\usepackage{stmaryrd}

\usepackage[normalem]{ulem}

\usepackage{listings}
\lstdefinelanguage[x86]{Assembler}{
	morekeywords={call, push, ret, mov, inc, je, inc, cmp},
	morecomment=[l]{\%}
}

\usepackage{tikz}
\usepackage[justification=centering]{caption}
\usetikzlibrary{chains,fit,shapes,patterns}
\usepackage{graphicx}

\def\Xrightarrow[#1]{\xrightarrow{#1}\!\!^*}
\def\XRightarrow[#1]{\xRightarrow{#1}\!\!^*}

\usepackage{cite}

\title{Reachability Analysis of Pushdown Systems with an Upper Stack}
\author{Adrien Pommellet\inst{1} \and Marcio Diaz\inst{1} \and Tayssir Touili\inst{2}}
\institute{LIPN and Université Paris-Diderot, France
	\and LIPN, CNRS, and Université Paris 13, France}

\begin{document}

	\maketitle
	\setcounter{footnote}{0}
	\let\thefootnote\relax\footnotetext{This work was partially funded by the FUI project Freenivi.}

	\begin{abstract}
		Pushdown systems (PDSs) are a natural model for sequential programs, but they can fail to accurately represent the way an assembly stack actually operates. Indeed, one may want to access the part of the memory that is below the current stack or base pointer, hence the need for a model that keeps track of this part of the memory. To this end, we introduce pushdown systems with an upper stack (UPDSs), an extension of PDSs where symbols popped from the stack are not destroyed but instead remain just above its top, and may be overwritten by later push rules. We prove that the sets of successors $post^*$ and predecessors $pre^*$ of a regular set of configurations of such a system are not always regular, but that $post^*$ is context-sensitive, so that we can decide whether a single configuration is forward reachable or not. In order to under-approximate $pre^*$ in a regular fashion, we consider a bounded-phase analysis of UPDSs, where a phase is a part of a run during which either push or pop rules are forbidden. We then present a method to over-approximate $post^*$ that relies on regular abstractions of runs of UPDSs. Finally, we show how these approximations can be used to detect stack overflows and stack pointer manipulations with malicious intent.
		
		\keywords{pushdown systems, reachability analysis, stack pointer, finite automata}
	\end{abstract}

	\section{Introduction}
	
	Pushdown systems were introduced to accurately model the \emph{call stack} of a program. A \emph{call stack} is a stack data structure that stores information about the active procedures of a program such as return addresses, passed parameters and local variables. It is usually implemented using a \emph{stack pointer (sp)} register that indicates the head of the stack. Thus, assuming the stack grows downwards, when data is \emph{pushed} onto the stack, \emph{sp} is decremented before the item is placed on the stack. For instance, in $x86$ architecture \emph{sp} is decremented by $4$ (pushing $4$ bytes). When data is \emph{popped} from the stack, \emph{sp} is incremented. For instance, in $x86$ architecture \emph{sp} is incremented by $4$ (popping $4$ bytes).
	
	However, in a PDS, neither push nor pop rules are truthful to the assembly stack. During an actual pop operation on the stack, the item remains in memory and the stack pointer is increased, as shown in Figures \ref{fig:3-1a} and \ref{fig:3-1b}, whereas a PDS deletes the item on the top of the stack, as shown in Figures \ref{fig:3-1c} and \ref{fig:3-1d}.
	
	\begin{figure}
		\centering
		\begin{minipage}{.49\linewidth}
			\centering
			\begin{tikzpicture}
			\edef\sizetape{0.8cm}
			\tikzstyle{tmtape}=[draw,minimum size=\sizetape]
			\tikzstyle{tmhead}=[arrow box,draw,minimum size=.45cm,arrow box
			arrows={north:0.25cm}]
			
			\begin{scope}[start chain=1 going right,node distance=0mm, outer sep=0mm]
			\node [on chain=1,tmtape,draw=none]{$\ldots$};
			\node [on chain=1,tmtape,pattern=north west lines, pattern color=gray]{$a$};
			\node [on chain=1,tmtape](input){$b$};
			\node [on chain=1,tmtape]{$c$};
			\node [on chain=1,tmtape]{$d$};
			\node [on chain=1,tmtape,draw=none] {$\ldots$};
			\end{scope}
			
			\node [tmhead,yshift=-0.55cm] at (input.south){\textsf{sp}};
			\end{tikzpicture}
			\captionof{figure}{The original stack.}
			\label{fig:3-1a}
		\end{minipage}
		\begin{minipage}{.49\linewidth}
			\centering
			\begin{tikzpicture}
			\edef\sizetape{0.8cm}
			\tikzstyle{tmtape}=[draw,minimum size=\sizetape]
			\tikzstyle{tmhead}=[arrow box,draw,minimum size=.45cm,arrow box
			arrows={north:0.25cm}]
			
			\begin{scope}[start chain=1 going right,node distance=0mm, outer sep=0mm]
			\node [on chain=1,tmtape,draw=none]{$\ldots$};
			\node [on chain=1,tmtape,pattern=north west lines, pattern color=gray]{$a$};
			\node [on chain=1,tmtape,pattern=north west lines, pattern color=gray]{$b$};
			\node [on chain=1,tmtape](input){$c$};
			\node [on chain=1,tmtape]{$d$};
			\node [on chain=1,tmtape,draw=none] {$\ldots$};
			\end{scope}
			
			\node [tmhead,yshift=-0.55cm] at (input.south){\textsf{sp}};
			\end{tikzpicture}
			\captionof{figure}{The stack after one pop.}
			\label{fig:3-1b}
		\end{minipage}
	\end{figure}
	
	\begin{figure}
		\centering
		\begin{minipage}{.49\linewidth}
			\centering
			\begin{tikzpicture}
			\edef\sizetape{0.8cm}
			\tikzstyle{tmtape}=[draw,minimum size=\sizetape]
			
			\begin{scope}[start chain=1 going right,node distance=0mm, outer sep=0mm]
			\node [on chain=1,tmtape]{$b$};
			\node [on chain=1,tmtape]{$c$};
			\node [on chain=1,tmtape]{$d$};
			\node [on chain=1,tmtape,draw=none] {$\ldots$};
			\end{scope}
			\end{tikzpicture}
			\captionof{figure}{The original PDS stack.}
			\label{fig:3-1c}
		\end{minipage}
		\begin{minipage}{.49\linewidth}
			\centering
			\begin{tikzpicture}
			\edef\sizetape{0.8cm}
			\tikzstyle{tmtape}=[draw,minimum size=\sizetape]
			
			\begin{scope}[start chain=1 going right,node distance=0mm, outer sep=0mm]
			\node [on chain=1,tmtape]{$c$};
			\node [on chain=1,tmtape]{$d$};
			\node [on chain=1,tmtape,draw=none] {$\ldots$};
			\end{scope}
			\end{tikzpicture}
			\captionof{figure}{The PDS stack after one pop.}
			\label{fig:3-1d}
		\end{minipage}
	\end{figure}
	
	This subtle difference becomes important when we want to analyze programs that directly manipulate the stack pointer and use assembly code. Indeed, in most assembly languages, \emph{sp} can be used like any other register. As an example, the instruction $\mathsf{mov \;eax \;[sp - 4]}$ will put the value pointed to at address $\mathsf{sp-4}$ in the register \textsf{eax} (one of the general registers). Since $\mathsf{sp-4}$ is an address above the stack pointer, we do not know what is being copied into the register $\mathsf{eax}$, unless we have a way to record the elements that had previously been popped from the stack and not overwritten yet. Such instructions may happen in malicious assembly programs: malware writers tend to do unusual things in order to obfuscate their payload and thwart static analysis.
	
	\begin{figure}
		\centering
		\begin{minipage}{.49\linewidth}
			\centering
			\begin{tikzpicture}
			\edef\sizetape{0.8cm}
			\tikzstyle{tmtape}=[draw,minimum size=\sizetape]
			
			\begin{scope}[start chain=1 going right,node distance=0mm, outer sep=0mm]
			\node [on chain=1,tmtape,draw=none] {$\ldots$};
			\node [on chain=1,tmtape](bottom){$a$};
			\node [on chain=1,tmtape,right=3mm of bottom]{$b$};
			\node [on chain=1,tmtape]{$c$};
			\node [on chain=1,tmtape]{$d$};
			\node [on chain=1,tmtape,draw=none]{$\ldots$};
			\end{scope}
			\end{tikzpicture}
			\captionof{figure}{The original UPDS stacks.}
			\label{fig:3-1e}
		\end{minipage}
		\begin{minipage}{.49\linewidth}
			\centering
			\begin{tikzpicture}
			\edef\sizetape{0.8cm}
			\tikzstyle{tmtape}=[draw,minimum size=\sizetape]
			
			\begin{scope}[start chain=1 going right,node distance=0mm, outer sep=0mm]
			\node [on chain=1,tmtape,draw=none] {$\ldots$};
			\node [on chain=1,tmtape]{$a$};
			\node [on chain=1,tmtape](bottom) {$b$};
			\node [on chain=1,tmtape,right=3mm of bottom] {$c$};
			\node [on chain=1,tmtape]{$d$};
			\node [on chain=1,tmtape,draw=none]{$\ldots$};
			\end{scope}
			\end{tikzpicture}
			\captionof{figure}{The UPDS stacks after one pop.}
			\label{fig:3-1f}
		\end{minipage}
	\end{figure}
	
	Thus, it is important to record the part of the memory that is just above the stack pointer. To this end, we extend PDSs in order to keep track of this \emph{upper stack}: we introduce a new model called \emph{pushdown system with an upper stack} (UPDS) that extends the semantics of PDSs. In a UPDS, when a letter is popped from the top of the stack (\emph{lower stack} from now on), it is added to the bottom of a write-only \emph{upper stack}, effectively simulating the decrement of the stack pointer. This is shown in Figures \ref{fig:3-1e} and \ref{fig:3-1f}, where after being popped, $b$ is removed from the lower stack (on the right) and added to the upper stack (on the left) instead of being destroyed. The top of the lower stack and the bottom of the upper stack meet at the stack pointer.
	
	We prove that the following properties hold for the class of UPDSs:
	\begin{itemize}
		\item the sets of predecessors and successors of a regular set of configurations are not regular; however, the set of successors of a regular set of configurations is context-sensitive;
		
		\item the set of predecessors is regular given a limit of $k$ phases, a phase being a part of a run during which either pop or push rules are forbidden; this is an under-approximation of the actual set of predecessors;
		
		\item an over-approximation of the set of successors can be computed by abstracting the set of runs first;
	\end{itemize}
	
	We then show that the UPDS model and the approximations of its reachability sets can be used to find errors and security flaws in programs.
	
	This paper is the full version of \cite{PDT-lata17}.
	
	\smallskip
	\noindent
	{\bf Paper outline.}
	We define in Section 1 a new class of pushdown systems called \emph{pushdown systems with an upper stack}. We prove in Section 2 that neither the set of predecessors nor the set of successors of a regular set of configurations are regular, but that the set of successors is nonetheless context-sensitive. Then, in Section 3, we prove that the set of predecessors of an UPDS is regular given a \emph{bounded-phase} constraint. In Section 4, we give an algorithm to compute an over-approximation of the set of successors. In Section 5, we show how these approximations could be applied to find errors or security flaws in programs. Finally, we describe the related work in Section 6 and show our conclusion in Section 7.

	\section{Pushdown systems with an upper stack}
	
	\begin{definition}[pushdown system with an upper stack]
		A \emph{pushdown system with an upper stack} (UPDS) is a triplet $\mathcal{P} = ( P, \Gamma, \Delta )$ where $P$ is a finite set of control states, $\Gamma$ is a finite stack alphabet, and $\Delta \subseteq P \times \Gamma \times P \times ( \lbrace \varepsilon \rbrace \cup \Gamma \cup \Gamma^{2} )$ a finite set of transition rules.
	\end{definition}
	
	We further note $\Delta_{pop} = \Delta \cap P \times \Gamma \times P \times \lbrace \varepsilon \rbrace$, $\Delta_{switch} = \Delta \cap P \times \Gamma \times P \times \Gamma$, and $\Delta_{push} = \Delta \cap P \times \Gamma \times P \times \Gamma^2$. If $\delta = ( p, w, p', w') \in \Delta$, we write $\delta = ( p, w ) \rightarrow ( p', w' )$. In a UPDS, a write-only \emph{upper stack} is maintained above the stack used for computations (from then on called the \emph{lower stack}), and modified accordingly during a transition.
	
	For $x \in \Gamma$ and $w \in \Gamma^*$, $\left| w \right|_x$ stands for the number of times the letter $x$ appears in the word $w$, and $w^R$ for the mirror image of $w$. Let $\bar{\Gamma}$ be a disjoint copy (bijection) of the stack alphabet $\Gamma$. If $x \in \Gamma$ (resp. $\Gamma^*$), then its associated letter (resp. word) in $\bar{\Gamma}$ (resp. $\bar{\Gamma}^*$) is written $\bar{x}$.
	
	A \emph{configuration} of $\mathcal{P}$ is a triplet $\langle p, w_u, w_l \rangle$ where $p \in P$ is a control state, $w_u \in \Gamma^*$ an \emph{upper stack content}, and $w_l \in \Gamma^*$ a \emph{lower stack content}. Let $Conf_\mathcal{P} = P \times \Gamma^*$ be the set of configurations of $\mathcal{P}$.
	
	A set of configurations $\mathcal{C}$ of a UPDS $\mathcal{P}$ is said to be \emph{regular} if for all $p \in P$, there exists a finite-state automaton $\mathcal{A}_p$ on the alphabet $\bar{\Gamma} \cup \Gamma$ such that $\mathcal{L} ( \mathcal{A}_p ) = \lbrace \bar{w_u} w_l \mid \langle p, w_u, w_l \rangle \in \mathcal{C} \rbrace$, where $\mathcal{L} ( \mathcal{A} )$ stands for the language recognized by an automaton $\mathcal{A}$.
	
	From the set of transition rules $\Delta$, we can infer an \emph{immediate successor relation} $\rightarrow_\mathcal{P} = ( \underset{\delta \in \Delta}{\bigcup} \overset{\delta}{\rightarrow} )$ on configurations of $\mathcal{P}$, which is defined as follows:
	
	\begin{description}
		\item[Switch rules:] if $\delta = ( p, \gamma ) \rightarrow ( p', \gamma' ) \in \Delta_{switch}$, then $\forall w_u \in \Gamma^*$ and $\forall w_l \in \Gamma^*$, $\langle p, w_u, \gamma w_l \rangle \overset{\delta}{\rightarrow} \langle p', w_u, \gamma' w_l \rangle$. The top letter $\gamma$ of the lower stack is replaced by $\gamma'$, but the upper stack is left untouched (the stack pointer doesn't move).
		
		\item[Pop rules:] if $\delta = ( p, \gamma ) \rightarrow ( p', \varepsilon ) \in \Delta_{pop}$, then $\forall w_u \in \Gamma^*$ and $\forall w_l \in \Gamma^*$, $\langle p, w_u, \gamma w_l \rangle \overset{\delta}{\Rightarrow} \langle p', w_u \gamma, w_l \rangle$. The top letter $\gamma$ popped from the lower stack is added to the bottom of the upper stack (the stack pointer moves to the right), as shown in Figure \ref{fig:3-2a}.	
		
		\item[Push rules:] if $\delta = ( p, \gamma ) \rightarrow (p', \gamma' \gamma'' ) \in \Delta_{push}$, then $\forall w_l \in \Gamma^*$, $\forall w_u \in \Gamma^*$, $\langle p, \varepsilon, \gamma w_l \rangle \overset{\delta}{\Rightarrow} \langle p', \varepsilon, \gamma' \gamma'' w_l \rangle$ and $\forall x \in \Gamma$, $\langle p, w_u x, a w_l \rangle \overset{\delta}{\Rightarrow} \langle p', w_u, \gamma' \gamma'' w_l \rangle$. A new letter $b$ is pushed on the lower stack, and a single letter is deleted from the bottom of the upper stack in order to make room for it, unless the upper stack was empty (the stack pointer moves to the left), as shown in Figure \ref{fig:3-2b}.	
	\end{description}
	
	\begin{figure}
		\centering
		\begin{minipage}{.4\linewidth}
			\centering
			\begin{tikzpicture}
			\edef\sizetape{0.8cm}
			\tikzstyle{tmtape}=[draw,minimum size=\sizetape]
			
			\begin{scope}[start chain=1 going right,node distance=0mm, outer sep=0mm]
			\node [on chain=1,tmtape]{$\gamma_1$};
			\node [on chain=1,tmtape,draw=none]{\textbf{p}};
			\node [on chain=1,tmtape]{$\gamma$};
			\node [on chain=1,tmtape]{$\gamma_2$};
			\node [on chain=1,tmtape]{$\gamma_3$};
			\end{scope}
			\end{tikzpicture}
		\end{minipage}
		$\overset{\delta}{\rightarrow}$
		\begin{minipage}{.4\linewidth}
			\centering
			\begin{tikzpicture}
			\edef\sizetape{0.8cm}
			\tikzstyle{tmtape}=[draw,minimum size=\sizetape]
			
			\begin{scope}[start chain=1 going right,node distance=0mm, outer sep=0mm]
			\node [on chain=1,tmtape]{$\gamma_1$};
			\node [on chain=1,tmtape]{$\gamma$};
			\node [on chain=1,tmtape,draw=none]{\textbf{p'}};
			\node [on chain=1,tmtape]{$\gamma_2$};
			\node [on chain=1,tmtape]{$\gamma_3$};
			\end{scope}
			\end{tikzpicture}
		\end{minipage}
		\captionof{figure}{Semantics of pop rules.}
		\label{fig:3-2a}
	\end{figure}
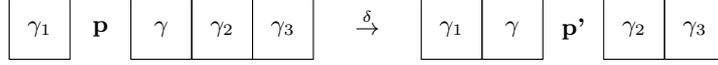
	
	\begin{figure}
		\centering
		\begin{minipage}{.4\linewidth}
			\centering
			\begin{tikzpicture}
			\edef\sizetape{0.8cm}
			\tikzstyle{tmtape}=[draw,minimum size=\sizetape]
			
			\begin{scope}[start chain=1 going right,node distance=0mm, outer sep=0mm]
			\node [on chain=1,tmtape]{$\gamma_1$};
			\node [on chain=1,tmtape]{$\gamma_2$};
			\node [on chain=1,tmtape,draw=none]{\textbf{p}};
			\node [on chain=1,tmtape]{$\gamma$};
			\node [on chain=1,tmtape]{$\gamma_3$};
			\end{scope}
			\end{tikzpicture}
		\end{minipage}
		$\overset{\delta}{\rightarrow}$
		\begin{minipage}{.4\linewidth}
			\centering
			\begin{tikzpicture}
			\edef\sizetape{0.8cm}
			\tikzstyle{tmtape}=[draw,minimum size=\sizetape]
			
			\begin{scope}[start chain=1 going right,node distance=0mm, outer sep=0mm]
			\node [on chain=1,tmtape]{$\gamma_1$};
			\node [on chain=1,tmtape,draw=none]{\textbf{p'}};
			\node [on chain=1,tmtape]{$\gamma'$};
			\node [on chain=1,tmtape]{$\gamma''$};
			\node [on chain=1,tmtape]{$\gamma_3$};
			\end{scope}
			\end{tikzpicture}
		\end{minipage}
		\captionof{figure}{Semantics of push rules.}
		\label{fig:3-2b}
	\end{figure}
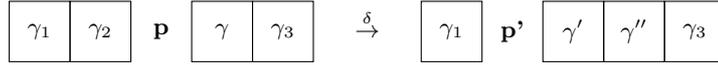
	
	The \emph{reachability} relation $\Rightarrow_\mathcal{P}$ is the reflexive and transitive closure of the immediate successor relation $\rightarrow_\mathcal{P}$. If $\mathcal{C}$ is a set of configurations, we introduce its set of \emph{successors} $post^* (\mathcal{P}, \mathcal{C} ) = \lbrace c \in P \times \Gamma^* \times \Gamma^* \mid \exists c' \in \mathcal{C}, c' \Rightarrow_\mathcal{P} c \rbrace$ and its set of \emph{predecessors} $pre^* (\mathcal{P}, \mathcal{C} ) = \lbrace c \in P \times \Gamma^* \times \Gamma^* \mid \exists c' \in \mathcal{C}, c \Rightarrow_\mathcal{P} c' \rbrace$. We may omit the variable $\mathcal{P}$ when only a single UPDS is being considered.
	
	For a set of configurations $\mathcal{C}$, let $\mathcal{C}_{low} = \lbrace \langle p, w_l \rangle \mid \exists w_u \in \Gamma^*, \langle p, w_u, w_l \rangle  \in \mathcal{C} \rbrace$ and $\mathcal{C}_{up} = \lbrace \langle p, w_u \rangle \mid \exists w_l \in \Gamma^*, \langle p, w_u, w_l \rangle \in \mathcal{C} \rbrace$. We then define $post^*_{up} ( \mathcal{P}, \mathcal{C} )  = ( post^* ( \mathcal{P}, \mathcal{C} ) )_{up}$, as well as $post^*_{low} ( \mathcal{P}, \mathcal{C} )$, $pre^*_{up} ( \mathcal{P}, \mathcal{C} )$ and $pre^*_{low} ( \mathcal{P}, \mathcal{C} )$ in a similar fashion.
	
	A \emph{finite run} $r$ of $\mathcal{P}$ from a configuration $c \in Conf_\mathcal{P}$ is a finite sequence of configurations $(c_i)_{i = 0, \dots, n}$ such that $c_0 = c$ and $c_0 \overset{t_1}{\rightarrow} c_1 \overset{t_2}{\rightarrow} c_2 \dots \overset{t_n}{\rightarrow} c_n$, where $t = (t_i)_{i = 1, \dots, n}$ is a sequence of transitions in $\Delta^*$, also called the \emph{trace} of $r$. We then write $c_0 \overset{t}{\Rightarrow}_{\mathcal{P}} c_n$, or $c_0 \Rightarrow^n_{\mathcal{P}} c_n$ ($c_n$ is reachable from $c_0$ in $n$ steps).
	
	We say that $r$ is a run of $\mathcal{P}$ from a set of configurations $\mathcal{C}$ if and only if $\exists c \in \mathcal{C}$ such that $r$ is a run of $\mathcal{P}$ from $c$. Let $Runs (\mathcal{P}, \mathcal{C})$ (resp. $Traces (\mathcal{P}, \mathcal{C})$) be the set of all finite runs (resp. traces) of $\mathcal{P}$ from a set of configurations $\mathcal{C}$.
	
	A UPDS and a PDS indeed share the same definition, but the semantics of the former expand the latter's. For a set $\mathcal{C} \subseteq P \times \Gamma^*$ of lower stack configurations (the upper stack is ignored) and a UPDS $\mathcal{P}$, let $post_{\text{\tiny{\emph{PDS}}}}^* ( \mathcal{P}, \mathcal{C} )$ and $pre_{\text{\tiny{\emph{PDS}}}}^* ( \mathcal{P}, \mathcal{C} )$ be the set of forward and backward reachable configurations from $\mathcal{C}$ using the PDS semantics. The following lemmas hold:
	
	\begin{lemma}
		\label{lm_runs_pds_upds}
		Given a UPDS $\mathcal{P} = ( P, \Gamma, \Delta )$ and a set of configurations $\mathcal{C}$, $t$ is a trace from $\mathcal{C}$ with respect to the UPDS semantics if and only if $t$ is a trace from $\mathcal{C}_{low}$ with respect to the standard PDS semantics.
	\end{lemma}
	
	\begin{lemma}
		\label{lm_reach_pds_upds}
		Given a UPDS $\mathcal{P} = ( P, \Gamma, \Delta )$ and a set of configurations $\mathcal{C}$,  $post^*_{low} \linebreak ( \mathcal{P}, \mathcal{C} ) = post_{\text{\tiny{PDS}}}^* ( \mathcal{P}, \mathcal{C}_{low} )$ and $pre^*_{low} ( \mathcal{P}, \mathcal{C} ) = pre_{\text{\tiny{PDS}}}^* ( \mathcal{P}, \mathcal{C}_{low} )$.
	\end{lemma}
	
	Lemmas \ref{lm_runs_pds_upds} and \ref{lm_reach_pds_upds} are true because, if we ignore the upper stack, a PDS and a UPDS share the same semantics.

	\section{Reachability properties}
	
	As shown in \cite{BEM-concur97, EHRS-cav00}, we know that $pre_{\text{\tiny{\emph{PDS}}}}^*$ and $post_{\text{\tiny{\emph{PDS}}}}^*$ are regular for a regular set of starting configurations. We prove that these results cannot be extended to UPDSs, but that $post^*$ is still context-sensitive. This implies that reachability of a single configuration is decidable for UPDSs.

	\subsection{$post^*$ is not regular}
	
	The following counterexample proves that, unfortunately, $post^* ( \mathcal{P}, \mathcal{C} )$ is not always regular for a given regular set of configurations $\mathcal{C}$ and a UPDS $\mathcal{P}$. The intuition behind this statement is that the upper stack can be used to store symbols in a non-regular fashion. The counter-example should be carefully designed in order to prevent later push operations from overwriting these symbols.
	
	Let $\mathcal{P} = ( P, \Gamma, \Delta )$ be a UPDS with $P = \lbrace p, p' \rbrace, \Gamma = \lbrace a, b, x, y, \bot \rbrace$, and $\Delta$ the following set of pushdown transitions:
	\[
	\begin{array}{cccc}
	( S_x ) & ( p, x ) \rightarrow ( p, a ) & ( R_a ) & ( p, a ) \rightarrow ( p, \varepsilon ) \\
	( S_y ) & ( p, y ) \rightarrow ( p, b ) & ( R_b ) & ( p, b ) \rightarrow ( p, \varepsilon ) \\
	( C ) & ( p, a ) \rightarrow ( p, ab ) & ( E ) & ( p, \bot ) \rightarrow ( p', \bot )
	\end{array}
	\]
	Let $\mathcal{C} = \lbrace p \rbrace \times \lbrace \varepsilon \rbrace \times x ( yx )^* \bot$ be a regular set of configurations. We can compute a relevant subset $L$ of $post^* ( \mathcal{C} )$:
	
	\begin{lemma}
		\label{lm_post_not_reg_1}
		$L = \lbrace \langle p', a^{n+1} b^{n}, \bot \rangle, n \in \mathbb{N} \rbrace \subseteq \text{post}^* ( \mathcal{C} )$.
	\end{lemma}
	
	\begin{proof}
		We prove that $\langle p, \varepsilon, x(yx)^{n} \bot \rangle \Rightarrow \langle p, a^{n+1} b^{n}, \bot \rangle$ by induction on $n$.
		\begin{description}
			\item[Basis:] $\langle p, \varepsilon, x \bot \rangle \Rightarrow \langle p, \varepsilon, a \bot \rangle \Rightarrow \langle p, a, \bot \rangle$.
			
			\item[Induction step:] if $\langle p, \varepsilon, ( xy )^{n} x \bot \rangle \Rightarrow \langle p, a^{n+1} b^{n}, \bot \rangle$, since the only rule able to read or modify the symbol $\bot$ is $( E )$ but it has not been applied as the PDS would end up in state $p'$, we have $\langle p, \varepsilon, ( xy )^{n} x \rangle \Rightarrow \langle p, a^{n+1} b^{n}, \varepsilon \rangle$, hence, $\langle p, \varepsilon, ( xy )^{n+1} x \bot \rangle \Rightarrow \langle p, a^{n+1} b^{n}, yx \bot \rangle$.
			
			However, $\langle p, a^{n+1} b^{n}, yx \bot \rangle \overset{S_y R_b S_x C^{n+1}}{\Rightarrow} \langle p, a^{n+1}, ab^{n+1} \bot \rangle$ and also $\langle p, a^{n+1}, \linebreak  ab^{n+1} \bot \rangle \overset{R_a R_b^{n+1}}{\Rightarrow} \langle p, a^{n+2}b^{n+1}, \bot \rangle$. From there, we have $\langle p, a^{n+1} b^{n}, \bot \rangle \overset{E}{\Rightarrow} \langle p', a^{n+1} b^{n}, \bot \rangle$.
		\end{description}
		Hence, $\langle p', a^{n+1} b^{n}, \bot \rangle \in post^* (C), \forall n \in \mathbb{N}$.
	\end{proof}
	
	Then, we prove an inequality that holds for any configuration in $post^*$:
	
	\begin{lemma}
		\label{lm_post_not_reg_2}
		$\forall \langle p, w_u, w_l \rangle \in \text{post}^* ( \mathcal{C} )$, $w = \bar{w_u} w_l$, $\left| w \right|_b + \left| w \right|_{\bar{b}} + 1 \geq \left| w \right|_a + \left| w \right|_{\bar{a}}$.
	\end{lemma}
	
	\begin{proof}
		The only rule in $\Delta$ that can add a letter $a$ to the whole stack is $S_x$. However, in order to apply it more than once, a $x$ deeper in the lower stack must be reached beforehand, and the only way to do so is by switching a $y$ to a $b$ and popping said $b$, hence, adding a $b$ to the whole stack.
		
		Moreover, the number of $b$ in the whole stack keeps growing during a computation, since no rule can switch a $b$ on the lower stack or delete it from the upper stack. The inequality therefore holds.
	\end{proof}
	
	If we suppose that $post^* ( \mathcal{C} )$ is regular, then so is the language $L^{p'}$, where $L^{p'}= \lbrace \bar{w_u} w_l \mid \langle p', w_u, w_l \rangle \in post^* ( \mathcal{C} ) \rbrace$, and by the pumping lemma, it admits a pumping length $k$. We want to apply the pumping lemma to an element of $L$ in order to generate a configuration that should be in $post^*$ but does not comply with the previous inequality.
	
	According to Lemma \ref{lm_post_not_reg_1}, $L \subseteq post^* ( \mathcal{C} )$ and as a consequence the word $w = \overline{a^{k+1} b^{k}} \bot$ is in $L^{p'}$. Hence, if we apply the pumping lemma to $w$, there exist $x, y, z \in ( \Gamma \cup \bar{\Gamma} )^*$ such that $w = xyz$, $\left| xy \right| \leq k$, $\left| y \right| \geq 1$, and $xy^iz \in post^* ( \mathcal{C} )$, $\forall i \geq 1$. As a consequence of $w$'s definition, $x, y \in \bar{a}^*$ and $z \in ( \bar{a} + \bar{b} )^*$.
	
	Hence, for $i$ large enough, $w_i = xy^iz \in L^{p'}$ and $\left| w_i \right|_{\bar{a}} > \left| w_i \right|_{\bar{b}} +1$. By Lemma \ref{lm_post_not_reg_2}, this cannot happen and therefore neither $L^{p'}$ nor $post^* ( \mathcal{C} )$ are regular.
	
	It should be noted that $L_{up}^{p'}$ is not regular either. Indeed, from the definition of $\mathcal{P}$ and $\mathcal{C}$, it is clear that $\forall \langle p', w_u, w_l \rangle \in post^* ( \mathcal{C} ), w_l = \bot$, so $L_{up}^{p'}$ and $L^{p'}$ are in bijection. We have therefore proven the following theorem:
	
	\begin{theorem}
		\label{theo_upds_post_not_reg}
		There exist a UPDS $\mathcal{P}$ and a regular set of configurations $\mathcal{C}$ for which neither $post^* ( \mathcal{C} )$ nor $post^*_{up} ( \mathcal{C} )$ are regular.
	\end{theorem}

	\subsection{$pre^*$ is not regular}
	
	We now prove that $pre^*$ is not regular either. Let $\mathcal{P} = ( P, \Gamma, \Delta )$ be a UPDS with $P = \lbrace p \rbrace, \Gamma = \lbrace a, b, c \rbrace$, and $\Delta$ the following set of pushdown transitions:
	\[
	\begin{array}{cccc}
	( C_0 ) & ( p, c ) \rightarrow ( p, ab ) & ( R_a ) & ( p, a ) \rightarrow ( p, \varepsilon ) \\
	( C_1 ) & ( p, c ) \rightarrow ( p, cb ) & ( R_b ) & ( p, b ) \rightarrow ( p, \varepsilon )
	\end{array}
	\]
	We define the regular set of configurations $\mathcal{C} = \lbrace p \rbrace \times ( ab )^* \times \lbrace c\rbrace$ and again, compute a relevant subset of $pre^* ( \mathcal{C} )$:
	
	\begin{lemma}
		\label{lm_pre_not_reg_1}
		$L = \lbrace \langle p, b^n, c^n c \rangle, n \in \mathbb{N} \rbrace \subseteq pre^* ( \mathcal{C} )$.
	\end{lemma}
	
	\begin{proof}
		By induction on $n$, we can prove that $\langle p, b^n, c^n c \rangle \Rightarrow \langle p, ( ab )^n, c \rangle$, proving the induction step by using the fact that $\langle p, b^{n+1}, c^{n+2} \rangle \Rightarrow \langle p, ab b^n, c^n c \rangle$.
	\end{proof}
	
	Given the rules of $\mathcal{P}$, the following lemma is verified:
	
	\begin{lemma}
		\label{lm_pre_not_reg_2}
		If $\langle p, b^m, c^n \rangle \Rightarrow^* \langle p, w_u, w_l \rangle$, then $\left| w_u \right|_a + \left| w_l \right|_a \leq n$.
	\end{lemma}
	
	\begin{proof}
		The only rule that can add an $a$ to the whole stack is $C_0$ and it replaces a $c$ on the lower stack by $ab$. Hence, during a computation, one cannot create more $a$ than there were $c$ in the initial configuration. The inequality therefore holds.
	\end{proof}
	
	If $pre^* ( \mathcal{C} )$ is regular, so is $L^{p} = \lbrace \bar{w_u} w_l \mid \langle p, w_u, w_l \rangle \in pre^* ( \mathcal{C} ) \rbrace$, and by the pumping lemma, it admits a pumping length $k$. Moreover, by lemma $\ref{lm_pre_not_reg_1}$, $w = \bar{b^k} c^k c \in L^{p}$.
	
	If we apply the pumping lemma to $w$, there exist $x, y, z \in ( \Gamma \cup \bar{\Gamma} )^*$ such that $w = xyz$, $\left| xy \right| \leq k$, $\left| y \right| \geq 1$ and $w_i = xy^iz \in pre^* ( \mathcal{C} )$, $\forall i \geq 1$. As a consequence of $w$'s definition, $x, y \in \bar{b}^*$ and $z \in \bar{b}^* c^k c$.
	
	Since $w_i \in L^{p}$, $\forall i \geq 1$, there exists an integer $n_i$ such that $w_i \Rightarrow c_i = \overline{( ab )^{n_i}} c$. Moreover, the size of the stack must grow or remain constant during a computation, hence $\left| c_i \right| \geq \left| w_i \right|$ and $n_i \geq \frac{\left| w_i \right| - 1}{2}$. Since words in the sequence $( w_i )_i$ are unbounded in length, the sequence $( n_i )_i$ must be unbounded as well. However, by Lemma \ref{lm_pre_not_reg_2}, $n_i = \left| c_i \right|_{\bar{a}} \leq \left| w_i \right|_c = k + 1$.
	
	Hence, there is a contradiction and $pre^* ( \mathcal{C} )$ is not regular. 
	
	\begin{theorem}
		\label{theo_upds_pre_not_reg}
		There exist a UPDS $\mathcal{P}$ and a regular set of configurations $\mathcal{C}$ for which $pre^* ( \mathcal{C} )$ is not regular.
	\end{theorem}

	\subsection{$post^*$ is context-sensitive}
	
	We prove that, if $\mathcal{C}$ is a regular set of configurations of a UPDS $\mathcal{P}$, then $post^* ( \mathcal{P}, \mathcal{C} )$ is context-sensitive. This implies that we can decide whether a single configuration is reachable from $\mathcal{C}$ or not.
	
	We show that the problem of computing $post^* ( \mathcal{P}, \mathcal{C} )$ can be reduced without loss of generality to the case where $\mathcal{C}$ contains a single configuration. To do so, we define a new UPDS $\mathcal{P'}$ by adding new states and rules to $\mathcal{P}$ such that any configuration $c$ in $\mathcal{C}$ can be reached from a single configuration $c_\$ = \langle p_\$, \varepsilon, \$ \rangle$. Once a configuration in $\mathcal{C}$ is reached, $\mathcal{P'}$ follow the same behaviour as $\mathcal{P}$.
	
	\begin{theorem}
		\label{theo_upds_post_simple}
		For each UPDS $\mathcal{P} = ( P, \Gamma, \Delta )$ and each regular set of configurations $\mathcal{C}$ on $ \mathcal{P}$, there exists a UPDS $\mathcal{P'} = ( P', \Gamma \cup \bar{\Gamma} \cup \lbrace \$ \rbrace, \Delta' )$, $P \subseteq P'$, and $p_\$ \in P' \setminus P$ such that $post^* ( \mathcal{P}, \mathcal{C} ) = post^* ( \mathcal{P'}, \lbrace \langle p_\$, \varepsilon, \$ \rangle \rbrace ) \cap ( P \times \Gamma^* \times \Gamma^* )$.
	\end{theorem}
	
	\begin{proof}
		Our intuition is to build configurations in $\mathcal{C}$ in three steps: from $c_\$$, push the word $\bar{w_u} w_l$ on the stack by using push rules mimicking a finite automaton accepting the regular set $\mathcal{C}$, switch each symbol in $\bar{\Gamma}$ to its equivalent letter in $\Gamma$ and then pop it in order to write $w_u$ on the upper stack, then move to the right state $p$.
		
		Since $\mathcal{C}$ is regular, so is $\forall p \in p$ the language $\lbrace \bar{w_u} w_l \mid \langle p, w_u, w_l \rangle \in \mathcal{C} \rbrace$. Consider $\mathcal{A}_p = ( \Gamma \cup \bar{\Gamma}, Q_p, E_p, I_p, F_p )$ such that $\mathcal{L} ( \mathcal{A}_p ) = \lbrace ( \bar{w_u} w_l )^R \mid \langle p, w_u, w_l \rangle \in \mathcal{C} \rbrace$. The mirror image is needed because the bottom of lower stack should be pushed first and the top of upper stack last. Without loss of generality, we suppose that $I_p = \lbrace i_p \rbrace$, $F_p = \lbrace f_p \rbrace$, $p_{\$} \notin Q_p$, $Q_p \cap P = \emptyset$ and that no edge in $E_p$ ends in $i_p$ nor starts in $f_p$.
		
		We define the UPDS $\mathcal{P'}_p = ( Q_p \cup \lbrace p_{\$}, p, p_\tau \rbrace, \Gamma \cup \bar{\Gamma} \cup \lbrace \$ \rbrace, \Delta'_p )$, where $p_\tau \notin Q_p$ and the following rules belong to $\Delta'_p$:
		\begin{description}
			\item[Rules from $\mathcal{A}_p$:] for all $x \in \Gamma \cup \bar{\Gamma}$, if $q \overset{x}{\rightarrow^*_{E_p}} q'$ in the automaton $\mathcal{A}_p$, then $(p_\$, \$ ) \rightarrow ( q, x ) \in \Delta'_p$ if $q = i_p$, and $y \in \Gamma \cup \bar{\Gamma}$, $(q, y ) \rightarrow ( q', xy ) \in \Delta'_p$ otherwise. These rules are used to build the stack and mimicks transitions in $\mathcal{A}_p$; symbols that will end on the upper stack are stored on the lower stack.
			
			\item[Setting the upper stack:] for all $\bar{x} \in \bar{\Gamma}$, $( f_p, \bar{x} ) \rightarrow ( p_\tau, x ) \in \Delta'_p$ and $( p_\tau, x ) \rightarrow ( f_p, \varepsilon ) \in \Delta'_p$. Each symbol $\bar{x}$ on the top of the lower stack is switched to its equivalent symbol in $\Gamma$ then popped in order to end on the upper stack.
			
			\item[Moving to state $p$:] for all $x \in \Gamma$, $( f_p, x ) \rightarrow ( p, x ) \in \Delta'_p$. Once the upper stack has been defined and the lower stack is being read, the UPDS moves to state $p$ in order to end in a configuration in $\mathcal{C}$.
		\end{description}
		
		The UPDS $\mathcal{P'}_p$ follows the three steps previously outlined: push $\bar{w_u} w_l$ on the lower stack, so that it is in a configuration $\langle f_p, \varepsilon, \bar{w_u} w_l \rangle$, then move to $\langle p_\tau, w_u, w_l \rangle$ by switching and popping the symbols in $\bar{\Gamma}$, and end in $\langle p, w_u, w_l \rangle$.
		
		We then introduce $\mathcal{P'} = ( \underset{p \in P}{\bigcup} P'_p \cup P, \Gamma \cup \bar{\Gamma} \cup \lbrace \$ \rbrace, \underset{p \in P}{\bigcup} \Delta'_p \cup \Delta )$. From $c_\$$, $\mathcal{P'}$ can reach any configuration of $\mathcal{C}$ using the rules of the automata $( \mathcal{P'}_p )_{p \in P}$, then follow the rules of $\mathcal{P}$. $\mathcal{P'}$ therefore satisfies Theorem \ref{theo_upds_post_simple}.
	\end{proof}
	
	Using this theorem, we can focus on the single starting configuration case. We assume for the rest of this subsection that $\mathcal{C} = \lbrace c_\$ \rbrace$, $p_\$ \in P$, and $\$ \in \Gamma$. We now formally define context-sensitive grammars:
	
	\begin{definition}
		A \emph{grammar} $\mathcal{G}$ is a quadruplet $( N, \Sigma, R, S )$ where $N$ is a finite set of non terminal symbols, $\Sigma$ a finite set of terminal symbols with $N \cap \Sigma = \emptyset$, $R \subseteq ( N \cup \Sigma )^* N ( N \cup \Sigma )^* \times ( N \cup \Sigma )^*$ a finite set of production rules, and $S \in N$ a start symbol.
	\end{definition}
	
	We define the one-step derivation relation $\dashrightarrow_\mathcal{G}$ on a given grammar $\mathcal{G}$: if $\exists p, q \in ( N \cup \Sigma )^*$, $p \rightarrow q \in R$ then $\forall x = upv, y = uqv \in ( N \cup \Sigma )^*$, $x \dashrightarrow_\mathcal{G} y$. The derivation relation $\dashrightarrow_\mathcal{G}^*$ is its transitive closure. The language $\mathcal{L} ( \mathcal{G} )$ of a grammar is the set $\lbrace w \in \Sigma^* \mid S\dashrightarrow_\mathcal{G}^* w \rbrace$. We may omit the variable $\mathcal{G}$ when only a single grammar is being considered.
	
	A grammar is said to be \emph{context-sensitive} if each productions rule $r \in R$ is of the form $\alpha A \beta \rightarrow \alpha \gamma \beta$ where $ \alpha, \beta \in ( N \cup \Sigma )^*, \gamma \in ( N \cup \Sigma )^+$, and $A \in N$. A language $L$ is said to be context-sensitive if there exists a context-sensitive grammar $\mathcal{G}$ such that $\mathcal{L} ( \mathcal{G} ) = L$.
	
	The following theorem is a well-know property of context-sensitive languages detailed in \cite{hopcroft-2000}:
	\begin{theorem}
		\label{theo_contsens_dec}
		Given a context-sensitive language $L$ and a word $w \in \Sigma^*$, we can effectively decide whether $w \in L$ or not.
	\end{theorem} 
	
	We can compute a context-sensitive grammar recognizing $post^*$. Our intuition is to represent a configuration $\langle p, w_u, w_l \rangle$ of $\mathcal{P}$ by a word $\top w_u p w_l \bot$ of a grammar $\mathcal{G}$. We use Theorem \ref{theo_upds_post_simple} so that the single start symbol of $\mathcal{G}$ can be matched to a single configuration $c_\$$. The context-sensitive rules of $\mathcal{G}$ mimic the transitions of the UPDS. As an example, a rule $\delta = ( p, a ) \rightarrow ( p', \varepsilon ) \in \Delta_{pop}$ can be modelled by three rules $p a \dashrightarrow_\mathcal{G} p g_\delta$, $ p g_\delta \dashrightarrow_\mathcal{G} a g_\delta$, and $a g_\delta \dashrightarrow_\mathcal{G} a p'$ such that $p a \dashrightarrow^*_\mathcal{G} a p'$, where $\dashrightarrow_\mathcal{G}$ stands for the one-step derivation relation and $g_\delta$ is a nonterminal symbol of $\mathcal{G}$.
	
	Let us define this context-sensitive grammar $\mathcal{G} = ( N, \Sigma, R, S )$ more precisely:
	\begin{description}
		\item[Start symbol:] $S$ is the start symbol.
		
		\item[Nonterminal symbols:] let $N = \left\{ S \right\} \cup \bar{\Gamma} \cup \bar{P} \cup \Delta_{switch} \cup \Delta_{pop} \cup \Delta_{push} \times \left\{ 0, 1 \right\}$. $\bar{P}$ is a disjoint copy (bijection) of the state alphabet $P$. In order to properly simulate transitions rules in $\Delta$ with context-sensitive production rules, nonterminal symbols related to these transitions are needed.
		
		\item[Terminal symbols:] $\Sigma = \lbrace \top, \bot \rbrace \cup P \cup \Gamma$.
		
		\item[Production rules:] $R = R_{\mathcal{P}} \cup R_{final} \cup \lbrace S \rightarrow \top \bar{p_\$} \bar{\$} \bot \rbrace$; the last rule initializes the starting configuration of $\mathcal{P}$.
	\end{description}
	
	The production rules in $R_\mathcal{P}$ simulate the semantics of the UPDS as defined by its transition rules $\Delta$. For each switch rule $\delta: ( p, a ) \rightarrow ( p', b ) \in \Delta_{switch}$, the following grammar rules belong to $R_\mathcal{P}$ in order to allow $p a \dashrightarrow^*_\mathcal{G} p' b$:
	\[
	\begin{array}{cccccc}
	( r_0^\delta ) & \bar{p} \bar{a} \rightarrow \delta \bar{a} & ( r_1^\delta ) & \delta \bar{a} \rightarrow \delta \bar{b} & ( r_f^\delta ) & \delta \bar{b} \rightarrow \bar{p'} \bar{b} 
	\end{array}
	\]
	
	For each pop rule $\delta: ( p, a ) \rightarrow ( p', \varepsilon ) \in \Delta_{pop}$, the following grammar rules belong to $R_\mathcal{P}$ in order to allow $p a \dashrightarrow^*_\mathcal{G} a p'$:
	\[
	\begin{array}{cccccc}
	( r_0^\delta ) & \bar{p} \bar{a} \rightarrow \bar{p} \delta & ( r_1^\delta ) & \bar{p} \delta \rightarrow \bar{a} \delta & ( r_f^\delta ) & \bar{a} \delta \rightarrow \bar{a} \bar{p'}
	\end{array}
	\]
	
	For each push rule $\delta: ( p, a ) \rightarrow ( p', bc ) \in \Delta_{push}$, the following grammar rules belong to $R_\mathcal{P}$ in order to allow $x p a \dashrightarrow^*_\mathcal{G} p' bc$ and $\top p a \dashrightarrow^*_\mathcal{G} \top p' bc$:
	\[
	\begin{array}{cccc}
	( r_0^\delta ) & \bar{p} \bar{a} \rightarrow \delta_0 \bar{a} & \forall x \in \Gamma, ( r_1^{\delta, x} ) & \bar{x} \delta_0 \rightarrow \delta_1 \delta_0 \\
	( r_1^{\delta, \top} ) & \top \delta_0 \rightarrow \top \delta_1 \delta_0 &
	( r_2^\delta ) & \delta_1 \delta_0 \bar{a} \rightarrow \delta_1 \delta_0 \bar{c} \\
	( r_3^\delta ) & \delta_1 \delta_0 \bar{c} \rightarrow \delta_1 \bar{b} \bar{c} & ( r_f^\delta ) & \delta_1 \bar{b} \bar{c} \rightarrow \bar{p'} \bar{b} \bar{c}
	\end{array}
	\]
	
	It is worth noting that, once a production rule $r_0^\delta$ has been applied, there is no other derivation possible in $R_{\mathcal{P}}$ but to apply the other production rules $( r_i^\delta )_i$ in the order they've been defined until a state symbol in $P$ has been written again by $r_f^\delta$. This sequence simulates a single transition rule of the UPDS $\mathcal{P}$.
	
	Finally, the rules in $R_{final}$ merely switch symbols in $\Gamma \cup P$ to their equivalent letters in $\Sigma$ in order to generate a terminal word, starting with the state symbol to prevent any further use of $R_\mathcal{P}$:
	\[
	\begin{array}{ccc}
	\forall p \in P & ( r_p^{final} ) & \bar{p} \rightarrow p \\
	\forall x, y \in \Gamma \cup P & ( r_{\bar{x},y}^{final} ) & bar{x} y \rightarrow x y \\
	\forall x, y \in \Gamma \cup P & ( r_{y,\bar{x}}^{final} ) & y \bar{x} \rightarrow y x
	\end{array}
	\]
	Once a rule $ r_p^{final}$ has been applied, the only production rules available for further derivations are in $R_{final}$.
	
	We prove that $\mathcal{L} ( \mathcal{G} )$ is in bijection with $post^* ( \lbrace c_\$ \rbrace )$.
	
	\begin{lemma}
		\label{lm_post_dec_1}
		If $\langle p, w_u, w_l \rangle \in post^{*} ( \lbrace c_\$ \rbrace )$, then from $S$ we can derive the nonterminal word $\top \bar{w_u} \bar{p} \bar{w_l} \bot$ in $\mathcal{G}$, and $\top w_u p w_l \bot \in \mathcal{L} ( \mathcal{G} )$.
	\end{lemma}
	
	\begin{proof} By induction on $n$, we must prove that if $c_\$ \Rightarrow^n \langle p, w_u, w_l \rangle$, then we can derive in $\mathcal{G}$ the nonterminal word $\top \bar{w_u} \bar{p} \bar{w_l} \bot$.
		\begin{description}
			\item[Basis:] we have $S \rightarrow \top \bar{p_\$} \bar{\$} \bot$.
			
			\item[Induction step:] if $c \Rightarrow^n \langle p, w_u, w_l \rangle \overset{\delta}{\Rightarrow} \langle p', w'_u, w'_l \rangle$, then the nonterminal word $\top \bar{w_u} \bar{p} \bar{w_l} \bot$ can be derived from $S$ in $\mathcal{G}$ by the induction hypothesis. From this word, we can further derive $\top \bar{w'_u} \bar{p'} \bar{w'_l} \bot$ using the production rules $r_0^\delta, \dots, r_f^\delta$ associated with the transition rule $\delta$.
		\end{description}
		Finally, from any non terminal word of the form $\top \bar{w_u} \bar{p} \bar{w_l} \bot$, we can derive in $\mathcal{G}$ a terminal $\top w_u p w_l \bot$ using rules in $R_{final}$.
	\end{proof}
	
	Moreover, by design of the grammar $\mathcal{G}$, the following lemma holds:
	
	\begin{lemma}
		\label{lm_post_dec_2}
		If $S \dashrightarrow^* \top w_u p w_l \bot$, $w_u, w_l \in \Gamma^*$, $p \in P$, then $\langle p, w_u, w_l \rangle \in post^* ( \lbrace c_\$ \rbrace )$. 
	\end{lemma}
	
	Hence, the following result holds:
	
	\begin{theorem}
		\label{theo_post_grammar}
		Given a UPDS $\mathcal{P}$ and a regular set of configurations $\mathcal{C}$, we can compute a context-sensitive grammar $\mathcal{G}$ such that $\langle p, w_u, w_l \rangle \in post^* ( \mathcal{P}, \mathcal{C} )$ if and only if $\top w_u p w_l \bot \in \mathcal{L} ( \mathcal{G} )$
	\end{theorem}
	
	Since the membership problem is decidable for context-sensitive grammars, the following theorem holds:
	
	\begin{theorem}
		\label{theo_post_decidable}
		Given a UPDS $\mathcal{P}$, a regular set of configurations $\mathcal{C}$, and a configuration $c$ of $\mathcal{P}$, we can decide whether $c \in post^* ( \mathcal{P}, \mathcal{C} )$ or not.
	\end{theorem}
	
	Unfortunately, this method cannot be extended to $pre^*$ due to a property of context-sensitive grammars: each time a context-sensitive rule is applied to a nonterminal word to produce a new word, the latter is of greater or equal length than the former. The forward reachability relation does comply with this monotony condition, as the combined size of the upper and lower stacks can only increase or stay the same during a computation, but the backward reachability relation does not.

	\section{Under-approximating $pre^*$}
	
	\emph{Under-approximations} of reachability sets can be used to discover errors in programs: if $\mathcal{X}$ is a regular set of forbidden configurations of a UPDS $\mathcal{P}$, $\mathcal{C}$ a regular set of starting configurations, and $U \subseteq pre^* ( \mathcal{X} )$ a regular under-approximation, then $U \cap \mathcal{C} \neq \emptyset$ implies that a forbidden configuration can be reached from the starting set, as shown in Figure \ref{fig:under_UPDS}. The emptiness of the above intersection has to be decidable, hence, the need for a regular approximation.
	
	\begin{figure}[!htb]
		\centering
		\begin{tikzpicture}
		\begin{scope}[blend group=soft light]
		\fill[red!30!white] (-3,0) circle (2);
		\fill[green!30!white] (-2.4,-0.1) circle (1.2);
		\fill[blue!30!white] (-0.3,0) circle (1.7);
		\end{scope}
		
		\node at (-3,1.5) {$pre^* (\mathcal{X})$};
		\node at (-2.6,0.4) {$\mathcal{U}$};
		\node at (0,1) {$\mathcal{C}$};
		\end{tikzpicture}
		\captionof{figure}{Using an under-approximation.}
		\label{fig:under_UPDS}
	\end{figure}
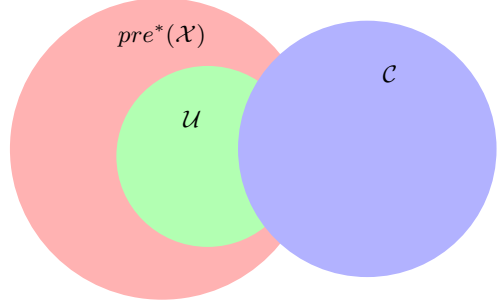

	\subsection{Multi-stack pushdown systems}
	
	Multi-stack pushdown systems (MPDSs) are pushdown systems with multiple stacks.
	
	\begin{definition}[La Torre et al. \cite{TMP-lics07}]
		A \emph{Multi-stack pushdown system} (or MPDS) is a quadruplet $\mathcal{M} = ( P, \Gamma, l, \Delta )$ where $P$ is a finite set of control states, $\Gamma$ is a finite stack alphabet, $l$ is the number of stacks, and $\Delta \subset P \times \Gamma \times \lbrace 1, \dots, l \rbrace \times P \times \Gamma^*$ a finite set of transition rules.
	\end{definition}
	
	For a given transition of a MPDS, in a given control state, only one stack is read and modified. A rule of the form $( p, w, n ) \rightarrow ( p', w' )$ is applied to the $n$-th stack with semantics similar to those of common pushdown systems.
	
	A \emph{configuration} of $\mathcal{M}$ is an element of $P \times ( \Gamma^* )^l$. A set of configurations $\mathcal{C}$ is said to be \emph{regular} if for all $p \in P$, there exists a finite-state automaton $\mathcal{A}_p$ on the alphabet $\lbrace \# \rbrace \cup \Gamma$ such that $\mathcal{L} ( \mathcal{A}_p ) = \lbrace w_1 \# \dots \# w_l \mid \left< p, w_1, \dots, w_l \right> \in \mathcal{C} \rbrace$.
	
	We define a successor relation $\hookrightarrow_\mathcal{M}$ on configurations. If $\delta = ( p, a, i ) \rightarrow ( p', w ) \in \Delta$, then for each configuration $c = \left< p, w_1, \dots, w_l \right>$ such that $w_i = ax$, we have $c \overset{\delta}{\hookrightarrow} \left< p', w'_1, \dots, w'_l \right>$ where $w'_i = wx$ and $w'_j = w_j$ if $j \neq i$. $\hookrightarrow^*_\mathcal{M}$ is the reflexive and transitive closure of the relation $\hookrightarrow_\mathcal{M} = ( \underset{\delta \in \Delta}{\bigcup} \overset{\delta}{\hookrightarrow} )$. We may ignore the variable $\mathcal{M}$ if only a single MPDS is being considered.
	
	For a given set of configurations $\mathcal{C}$ of a MPDS $\mathcal{M}$, we define its set of \emph{predecessors} $pre_{\text{{MPDS}}}^* ( \mathcal{M},\mathcal{C} ) = \lbrace c \in P \times ( \Gamma^* )^l \mid \exists c' \in \mathcal{C}, c \hookrightarrow^* c' \rbrace$.
	
	A \emph{run} $r$ of $\mathcal{M}$ from a configuration $c_0$ is a sequence of configurations $r = ( c_i )_{i = 1, \dots, n} \in \Delta^*$ such that $c_0 \overset{t_1}{\hookrightarrow} c_1 \overset{t_2}{\hookrightarrow} c_2 \dots \overset{t_n}{\hookrightarrow} c_n$, where $t = (t_i)_{i = 1, \dots, n}$ is a sequence of transition rules of $\mathcal{M}$ called the \emph{trace} of $r$. We then write $c_0 \overset{t}{\hookrightarrow} c_n$.
	
	Multi-stack automata are unfortunately Turing powerful, even with only two stacks. La Torre et al. thus introduced in \cite{TMP-lics07} a new restriction called \emph{phase-bounding}:
	
	\begin{definition}
		A configuration $c'$ of $\mathcal{M}$ is said to be reachable from another configuration $c$ in \emph{k phases} if there exists a sequence of runs $r_1, r_2, \dots r_{k}$ with matching traces $t_1, t_2, \dots t_{k}$ such that $c_0 \overset{t_1}{\hookrightarrow} c_1 \dots \overset{t_k}{\hookrightarrow} c_k$ where $c_0 = c$, $c_k = c'$, $(c_i)_{i = 1, \dots, k}$ is a sequence of configurations on $\mathcal{M}$ and where during the execution of a given run $r_i$, at most a single stack is popped from. We note $c \hookrightarrow^*_{\mathcal{M},k} c'$.
	\end{definition}
	
	We then define $pre_{\text{\tiny{MPDS}}}^* ( \mathcal{M},\mathcal{C}, k ) = \lbrace c \in P \times ( \Gamma^* )^l \mid \exists c' \in \mathcal{C}, c \hookrightarrow^*_{\mathcal{M},k} c' \rbrace$. The following theorem has been proven in \cite{seth-cav10}:
	
	\begin{theorem}
		\label{theo_mpds_pre_k_reg}
		Given a MPDS $\mathcal{M}$ and a regular set of configurations $\mathcal{C}$, the set $pre_{\text{\tiny{MPDS}}}^* ( \mathcal{M},\mathcal{C}, k )$ is regular and effectively computable.
	\end{theorem}

	\subsection{Application to UPDSs}
	
	The notion of bounded-phase computations can be extended to UPDSs. A run $r$ of $\mathcal{P}$ is said to be \emph{k-phased} if it is of the form: $r = r_1 \cdot r_2 \dots r_{k}$ where $\forall i \in \left\{ 1, \dots, k \right\}$, $r_i \in (\Delta_{Push} \cup \Delta_{Switch} )^* \cup ( \Delta_{Pop} \cup \Delta_{Switch} )^*$. During a phase, one can either push or pop, but can't do both. Such a run has therefore at most $k$ alternations between push and pop rules. We can extend this notion to traces.
	
	The \emph{k-bounded} reachability relation $\Rightarrow_k$ is defined as follows: $c_0 \Rightarrow_k c_1$ if there exists a $k$-phased run $r$ on $\mathcal{P}$ with a matching trace $t$ such that $c_0 \overset{t}{\Rightarrow} c_1$. Using this new reachability relation, given a set of configurations $\mathcal{C}$, we can define $pre^* ( \mathcal{P}, \mathcal{C}, k )$.
	
	We can show that a UPDS $\mathcal{P}$ can be simulated by a MPDS $\mathcal{M}$ with two stacks, the second stack of $\mathcal{M}$ being equivalent to the lower stack, and the first one, to a mirrored upper stack followed by a symbol $\bot$ that can't be popped and is used to know when the end of the stack has been reached. Elements of $P \times \Gamma^* \times \Gamma^*$ can equally be considered as configurations of $\mathcal{P}$ or $\mathcal{M}$, assuming in the latter case that we consider the mirror of the first stack and add a $\bot$ symbol to its bottom. Thus:
	
	\begin{lemma}
		For a given UPDS $\mathcal{P} = ( P, \Gamma, \Delta )$ and a regular set of configurations $\mathcal{C}$, there exists a MPDS $\mathcal{M}$, a regular set of configurations $\mathcal{C'}$, and $\bot \notin \Gamma$ such that $\langle p, w_u^R \bot, w_l \rangle \in pre_{\text{\tiny{MPDS}}}^* ( \mathcal{M},\mathcal{C'}, k ) \cap (P \times \Gamma^* \times \Gamma^* )$ if and only if $\langle p, w_u, w_l \rangle \in pre^* ( \mathcal{P}, \mathcal{C}, k )$.
	\end{lemma}
	
	\begin{proof}				
		Following the above intuition, we define a two-stack pushdown system $\mathcal{M} = ( P \cup \Delta_{push} \cup \Delta_{pop}, \Gamma \cup \lbrace \bot \rbrace, 2, \Delta' )$ where $\Delta'$ has the following rules:
		\begin{description}		
			\item[Switch rules:] if $\delta = ( p, a ) \rightarrow ( q, b ) \in \Delta_{switch}$, then $( p, a, 2 ) \rightarrow ( q, b ) \in \Delta'$.
			
			\item[Pop rules:] if $\delta = ( p, a ) \rightarrow ( q, \varepsilon ) \in \Delta_{pop}$, then $( p, a, 2 ) \rightarrow ( \delta, \varepsilon ) \in \Delta'$ and $( \delta, x, 1 ) \rightarrow ( q, ax ) \in \Delta'$ for each $x \in \Gamma \cup \lbrace \bot \rbrace$.
			
			\item[Push rules:] if $\delta = ( p, a ) \rightarrow ( q, bc ) \in \Delta_{push}$, then we define $( p, a, 2 ) \rightarrow ( \delta, bc ) \in \Delta'$, $( \delta, \bot, 1 ) \rightarrow ( q, \bot ) \in \Delta'$ and $( \delta, x, 1 ) \rightarrow ( q, \varepsilon ) \in \Delta'$ for each $x \in \Gamma$. If we reach $\bot$ on the second stack, the upper stack is considered to be empty and no symbol should be popped from it.
		\end{description}
		We then define $\mathcal{C'} = \lbrace \langle p, w_u^R \bot, w_l \rangle \mid \langle p, w_u, w_l \rangle \in \mathcal{C} \rbrace$. $\mathcal{M}$ simulates $\mathcal{P}$ and $\mathcal{C'}$ is equivalent to (in bijection with) $\mathcal{C}$.
	\end{proof}
	
	From Theorem \ref{theo_mpds_pre_k_reg}, we get:
	
	\begin{theorem}
		\label{theo_upds_pre_k_reg}
		Given a UPDS $\mathcal{P}$ and a regular set of configurations $\mathcal{C}$, the set $pre^* ( \mathcal{P}, \mathcal{C}, k )$ is regular and effectively computable.
	\end{theorem}
	
	$pre^* ( \mathcal{P}, \mathcal{C}, k )$ is then obviously an under-approximation of $pre^* ( \mathcal{P}, \mathcal{C} )$.

	\section{Over-approximating $post^*$}
	
	While under-approximations of reachability sets can be used to show that an error can occur, \emph{over-approximations} can, on the other hand, prove that a program is safe from a particular error. If $\mathcal{X}$ is a regular set of forbidden configurations on a UPDS $\mathcal{P}$, $\mathcal{C}$ a regular set of starting configurations, and $O \supseteq post^* ( \mathcal{C} )$ a regular over-approximation, then $O \cap \mathcal{X} = \emptyset$ implies that no forbidden configuration can be reached from the starting set and that the program is therefore safe, as shown in Figure \ref{fig:over_UPDS}.
	
	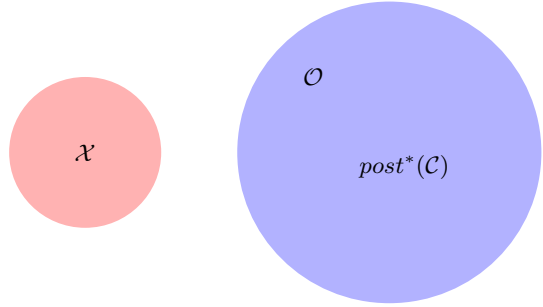
\begin{figure}[!htb]
		\centering
		\begin{tikzpicture}
		\begin{scope}[blend group=soft light]
		\fill[red!30!white] (-3,0) circle (1);
		\fill[green!30!white] (1.2,-0.2) circle (1.2);
		\fill[blue!30!white] (1,0) circle (2);
		\end{scope}
		
		\node at (-3,0) {$\mathcal{X}$};
		\node at (1.2,-0.2) {$post^* (\mathcal{C})$};
		\node at (0,1) {$\mathcal{O}$};
		\end{tikzpicture}
		\captionof{figure}{Using an over-approximation.}
		\label{fig:over_UPDS}
	\end{figure}
	
	The emptiness of the above intersection has to be decidable, hence, the need for a regular approximation.

	\subsection{A relationship between runs and the upper stack}
	
	We prove here that from a regular set of traces of a given UPDS, a regular set of corresponding upper stacks can be computed. A subclass of programs whose UPDS model has a regular set of traces would be programs with finite recursion (hence, with a stack of finite height).
	
	Given a UPDS $\mathcal{P} = ( P, \Gamma, \Delta )$ and a configuration $c = \langle p, w_u, w_l \rangle$ of $\mathcal{P}$, we match inductively to each sequence of transition $\tau \in \Delta^*$ an upper stack word $\upsilon ( \tau, c ) \in \Gamma^*$ according to the following rules:
	\begin{itemize}
		\item $\upsilon ( \varepsilon, c ) = w_u$;
		\item if $\delta =  ( p, \gamma ) \rightarrow ( p', \gamma' ) \in \Delta_{switch}$, then $\upsilon ( \tau \delta, c ) = \upsilon ( \tau, c )$;
		\item if $\delta =  ( p, \gamma ) \rightarrow ( p', \varepsilon ) \in \Delta_{pop}$, then $\upsilon ( \tau \delta, c ) = \upsilon ( \tau ) \cdot a$;
		\item if $\delta =  ( p, \gamma ) \rightarrow ( p', \gamma' \gamma'' ) \in \Delta_{push}$, then $\upsilon ( \tau \delta, c ) = \varepsilon$ if $\upsilon ( \tau, c ) = \varepsilon$, and $\upsilon ( \tau \delta, c ) = w$ if $\upsilon ( \tau ) = w x$, where $x \in \Gamma$;
	\end{itemize}
	
	Intuitively, $\upsilon ( \tau, c )$ is the upper stack content after applying the sequen\-ce of transitions $\tau$, starting from a configuration $c$. Note that $\tau$ may not be an actual trace of $\mathcal{P}$; $\upsilon ( \tau )$ is merely the \emph{virtual} upper stack built by pushing and popping values in a write-only manner, regardless of the lower stack, the control states, and the coherence of the sequence of transitions used. However, if $t$ is indeed a trace of $\mathcal{P}$, then the upper stack configuration $\upsilon ( t, c )$ is indeed reachable from $c$ using the trace $t$.
	
	A sequence of transitions is said to be \emph{meaningful} if $\forall p' \in P$, any transition ending in state $p'$ can only be followed by a transition starting in state $p'$. A trace of $\mathcal{P}$ is obviously a meaningful sequence. A set of sequences of transitions $T$ is said to be \emph{prefix-closed} if, given $t \in T$, any prefix of $t$ is in $T$ as well. The set of all traces of a given system is obviously prefix-closed.
	
	The following theorem holds:
	
	\begin{theorem}
		\label{theo_traces_reg_post}
		For a UPDS $\mathcal{P} = ( P, \Gamma, \Delta )$, a regular set of configurations $\mathcal{C}$, and a regular, prefix-closed set of meaningful sequences of transitions $T \subseteq \Delta^*$ of $\mathcal{P}$ from $\mathcal{C}$, the set of upper stack configurations $\mathcal{U} ( T ) = \lbrace \langle p', w'_u \rangle \mid \exists c = \langle p, w_u, w_l \rangle \in \mathcal{C}, \exists t \in T, t \text{ starts in state } p \text{ and ends in state } p', \upsilon ( t, c ) = w'_u \rbrace$ spawned by $T$ from $\mathcal{C}$ is regular and effectively computable.
	\end{theorem}
	
	Thanks to Theorem \ref{theo_upds_post_simple}, we consider the single configuration case where $\mathcal{C} = \lbrace c_\$ \rbrace$ without loss of generality. Let $\mathcal{A}_T = ( \Delta, Q, E, I, F )$ be a finite state automaton such that $\mathcal{L} ( \mathcal{A}_T ) = T$. Since $T$ is meaningful, we can assume that $Q = \underset{p \in P}{\cup} Q_p$ where $ Q_p$ is such that $\forall q \in Q_p$, if there is an edge $q' \xrightarrow{\delta}_E q$, then the pushdown rule $\delta$ is of the form $( p', a ) \rightarrow ( p, w )$. We can also assume that $F = Q$ since $T$ is prefix-closed.
	
	We introduce the automaton $\mathcal{A}_U = ( \Gamma, Q, E', I, F )$ whose set of transitions $E'$ is defined by applying the following rules until saturation, starting from $E' = \emptyset$:
	\begin{description}
		\item[$( S_{pop} )$] if there is an edge $q_0 \xrightarrow{\delta}_E q_1$ in $\mathcal{A}_T$ and $\delta$ is of the form $( p, a ) \rightarrow ( p', \varepsilon )$, then we add the edge $q_0 \xrightarrow{a} q_1$ to $E'$.
		
		\item[$( S_{switch} )$] if there is an edge $q_0 \xrightarrow{\delta}_E q_1$ in $\mathcal{A}_T$ and $\delta$ is of the form $( p, a ) \rightarrow ( p', b )$, then we add the edge $q_0 \xrightarrow{\varepsilon} q_1$ to $E'$.
		
		\item[$( S_{push} )$] if there is an edge $q_0 \xrightarrow{\delta}_E q_1$ in $\mathcal{A}_T$,  $\delta$ is of the form $( p, a ) \rightarrow ( p', bc )$, and there is a state $q$ such that either \textbf{(1)} $q \in Q$ and $q \xrightarrow{x}_{E'} q_0$ for $x \in \Gamma$ or \textbf{(2)} $q \in I$ and $q \Xrightarrow[\varepsilon]_{E'} q_0$, then we add an edge $q \xrightarrow{\varepsilon} q_1$ to $E'$.
	\end{description}
	We call $( E'_i )_i$ the finite, growing sequence of edges created during the saturation procedure.
	
	Our intuition behind the above construction is to create a new automaton that uses the states of the sequence automaton but accepts upper stack words instead: an upper stack word $w$ is accepted by $\mathcal{A}_U$ with the path $q_i \Xrightarrow[w]_{E'} q$,  where $q_i \in I$ if and only if $\mathcal{A}_T$ accepts a sequence $t$ with the path $q_i \Xrightarrow[t]_{E} q$ where $t$ ends in state $p$ and $\upsilon ( t, c_\$ ) = w$. This property is preserved at every step of the saturation procedure.
	
	Indeed, consider a sequence $t$ and $w = \upsilon ( t, c_\$ )$. Suppose that $t$ and $w$ satisfy the property above: there is a path $q_i \Xrightarrow[t]_{E} q_0$ in $\mathcal{A}_T$ and a path $q_i \Xrightarrow[w]_{E'} q_0$ in $\mathcal{A}_U$. Let $q_0 \xrightarrow{\delta}_E q_1$ be a transition of $\mathcal{A}_T$, $q_1 \in Q$ and $\delta \in \Delta$. $t \delta$ is a sequence of transitions in $T$ with a labelled path $q_i \Xrightarrow[t \delta]_{E} q_1$ in $\mathcal{A}_T$, and in order to satisfy the above property, a path $q_i \Xrightarrow[w']_{E'} q_1$ labelled by $w' = \upsilon ( t \delta, c_\$ )$ should exist in $\mathcal{A}_U$ as well.
	
	If $\delta \in \Delta_{pop}$, then $w' = \upsilon ( t \delta, c_\$ ) = w a$, where $a \in \Gamma$. Rule $( S_{pop} )$ creates an edge $q_0 \xrightarrow{a} q_1$ in $\mathcal{A}_U$ such that there is a path $q_i \Xrightarrow[w]_{E'} q_0 \xrightarrow{a}_{E'} q_1$ labelled by $w'$.
	
	If $\delta \in \Delta_{switch}$, then $w' = \upsilon ( t \delta, c_\$ ) = w$. Rule $( S_{switch} )$ creates an edge $q_0 \xrightarrow{\varepsilon} q_1$ in $\mathcal{A}_U$ such that there is a path $q_i \Xrightarrow[w]_{E'} q_0 \xrightarrow{\varepsilon}_{E'} q_1$ labelled by $w'$.
	
	If $\delta \in \Delta_{push}$ and $w = w_0 x$, where $x \in \Gamma$, then $w' = \upsilon ( t \delta, c_\$ ) = w_0$ and for every state $q \in Q$ such that $q_i \Xrightarrow[w_0]_{E'} q \xrightarrow{x}_{E'} q_0$, $( S_{push} )$ adds an edge $q \xrightarrow{\varepsilon} q_1$ to $\mathcal{A}_U$ such that there is a path $q_i \Xrightarrow[w_0]_{E'} q \xrightarrow{\varepsilon}_{E'} q_1$ labelled by $w'$.
	
	Following this intuition, we can prove this lemma:
	
	\begin{lemma}
		\label{lm_upp_ reg_1}
		For every sequence $t$ such that $\exists q_i \in I$, $\exists q \in Q_p$, $q_i \Xrightarrow[t]_E q$, then there exists a path $q_i \Xrightarrow[w]_{E'} q$ in $\mathcal{A}_U$ such that $\upsilon ( t, c_\$ ) = w$.
	\end{lemma}
	
	On the other hand, we must prove this lemma to get the full equivalence:
	
	\begin{lemma}
		\label{lm_upp_ reg_2}
		At any step $i$ of the saturation procedure, if $q_i \Xrightarrow[w]_{E'_i} q$ where $q_i \in I$, then there exists a sequence of transitions $t$ in $T$ such that $q_i \Xrightarrow[t]_E q$ and $\upsilon ( t, c_\$ ) = w$. 
	\end{lemma}
	
	\begin{proof}
		We prove this lemma by induction on the saturation step $i$:
		\begin{description}
			\item[Basis:] $E'_0 = \emptyset$ and the lemma holds.
			
			\item[Induction step:] Let $q_1 \Xrightarrow[x]_{E'_{i+1}} q_2$ be the $i + 1$-th transition added to $E'$. Let $w' = w x$ be such that there is a path  $q_i \Xrightarrow[w]_{E'_i} q_1 \xrightarrow{x}_{E'_{i+1}} q_2$. By induction hypothesis, there is a sequence $t \in T$ such that $q_i \Xrightarrow[t]_{E} q_1$ and $\upsilon ( t, c_\$ ) = w$.
			
			If $x \in \Gamma$, then there is a rule $\delta \in \Delta_{pop}$ popping $x$ from the stack such that $q_1 \xrightarrow{\delta}_{E} q_2$ by definition of the saturation rules. We have $\upsilon ( t \delta, c_\$ ) = w x = w'$ and the lemma holds at the $i + 1$-th step.
			
			If $x = \varepsilon$, then the rule $\delta$ spawning $q_1 \Xrightarrow[x]_{E'_{i+1}} q_2$ is either a switch or push rule. The switch case being similar to the pop case, we will consider that $\delta \in \Delta_{push}$.
			
			In the first case of the push saturation rule, there exists by definition a state $q_0$ such that $q_0 \xrightarrow{\delta}_E q_2$ and $y \in \Gamma$ such that $q_1 \xrightarrow{y}_{E'_i} q_0$. Hence, there is $\delta' \in \Delta_{pop}$ popping $y$ from the stack such that $q_1 \Xrightarrow[\delta']_{E} q_0$. If we consider the sequence $q_i \Xrightarrow[t]_{E} q_1 \xrightarrow{\delta'}_{E} q_0 \xrightarrow{\delta}_E q_2$, then $\upsilon ( t \delta' \delta, c_\$ ) = w = w'$ and the lemma holds at the $i + 1$-th step. The second case of the push saturation rule is similar.
		\end{description}
	\end{proof}
	
	Let $\mathcal{L}_p ( \mathcal{A}_U ) = \lbrace w \mid \exists q_i \in I, \exists f \in Q_p, i \Xrightarrow[w]_{E'} f \rbrace$ be the set of paths in $\mathcal{A}_U$ ending in a final node related to a state $p$ of $\mathcal{P}$. By Lemmas \ref{lm_upp_ reg_1} and \ref{lm_upp_ reg_2}, $\mathcal{U} ( R ) = \lbrace \langle p, w_u \rangle \mid w_u \in \mathcal{L}_p ( \mathcal{A}_U ) \rbrace$. Since the languages $\mathcal{L}_p$ are regular and there is a finite number of them, $\mathcal{U} ( T )$ is regular as well and can be computed using $\mathcal{A}_U$.

	\subsection{Computing an over-approximation}
	
	The set of traces of a UPDS $\mathcal{P} = ( P, \Gamma, \Delta )$ from a regular set of configurations $\mathcal{C}$ is not always regular. By Lemma \ref{lm_runs_pds_upds}, traces of $\mathcal{P}$ are the same for the UPDS and PDS semantics. Thus, we can apply methods originally designed for PDSs to over-approximate traces of a UPDS in a regular fashion, as shown in \cite{BS-90, PW-acl91}.
	
	With one of these methods, we can therefore compute a regular over approx\~imation $\mathcal{T} ( \mathcal{P}, \mathcal{C} )$ of the set of traces of $\mathcal{P}$ from $\mathcal{C}$. Using the saturation procedure underlying Theorem \ref{theo_traces_reg_post}, we can then compute the set $\mathcal{U} ( \mathcal{T} ( \mathcal{P}, \mathcal{C} ) )$ of upper stack configurations reachable using over-approximated traces of $\mathcal{P}$, hence, an over-approximation of the actual set of reachable upper stack configurations.
	
	However, we still lack the lower stack component of the reachability set. As shown in \cite{EHRS-cav00}, $post^*_{\text{\tiny{PDS}}} ( \mathcal{P}, \mathcal{C} )$ is regular and computable, and we can determine the exact set of reachable lower stack configurations.
	
	We define the set $O = \lbrace \langle p, w_u, w_t \rangle \mid \langle p, w_u \rangle \in \mathcal{T} ( \mathcal{R} ( \mathcal{P},\mathcal{C} ) ), \langle p, w_t \rangle \in post^*_{\text{\tiny{PDS}}} \linebreak ( \mathcal{P}, \mathcal{C} ) \rbrace$. $O$ is a regular over-approximation of $post^* ( \mathcal{P}, \mathcal{C} )$.

	\section{Applications}
	
	The UPDS model can be used to detect stack behaviours that cannot be found using a simple pushdown system. In this section, we present three such examples.

	\subsection{Stack overflow detection}
	
	A stack overflow is a programming malfunction occurring when the call stack pointer exceeds the stack bound. In order to analyze a program's vulnerability to stack overflow errors, we compute its representation as a UPDS $\mathcal{P} = ( P, \Gamma, \Delta )$, using the control flow model outlined in \cite{EHRS-cav00}.
	
	Let $\mathcal{C} = P \times \top \#^m \times L$ be the set of starting configurations, where $\top \in \Gamma$ is a top stack symbol that does not appear in any rule in $\Delta$, $\# \in \Gamma$ a filler symbol, $m$ an integer depending on the maximal size of the stack, and $L$ a regular language of lower stack initial words. Overwriting the top symbol would represent a stack overflow misfunction. Since there is no such thing as an upper stack in a simple pushdown automaton, we need a UPDS to detect this error, as shown in Figure \ref{fig:3-ex1}.
	
	\begin{figure}[!hth]
		\centering
		\begin{tikzpicture}
		\edef\sizetape{0.45cm}
		\tikzstyle{tmtape}=[draw,minimum size=1.5*\sizetape]
		
		\begin{scope}[start chain=1 going right,node distance=0mm, outer sep=0mm]
		\node [on chain=1,tmtape] (top) {$\top$};
		\node [on chain=1,tmtape]{$\#$};
		\node [on chain=1,tmtape,draw=none]{$\overset{m \text{ times}}{\ldots}$};
		\node [on chain=1,tmtape] (bottom) {$\#$};
		\node [on chain=1,tmtape,right=1.5mm of bottom] {$a$};
		\node [on chain=1,tmtape]{$b$};
		\node [on chain=1,tmtape,draw=none]{$\ldots$};
		\end{scope}
		\end{tikzpicture}
		\captionof{figure}{Using $\top$ to bound the upper stack.}
		\label{fig:3-ex1}
	\end{figure}
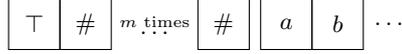
	
	Let $\mathcal{X} = P \times ( \Gamma \backslash \lbrace \top \rbrace )^* \times \Gamma^*$ be the set of forbidden configurations where the top stack symbol has been overwritten. If the intersection of the under-approximation $\mathcal{U}$ of $pre^* ( \mathcal{X} )$ with $\mathcal{C}$ is not empty, then a stack overflow does happen in the program. On the other hand, if the intersection of the over-approximation $\mathcal{O}$ of $post^* ( \mathcal{C} )$ with the set $\mathcal{X}$ of forbidden configurations is empty, then we are sure that a stack overflow will not happen in the program

	\subsection{Reading the upper stack}
	
	Let us consider the piece of code \ref{asm1}. In line 1, the bottom symbol of the upper stack $\mathsf{sp - 4}$, just above the stack pointer, is copied into the register \textsf{eax}. In line 2, the content of \textsf{eax} is compared to a given value $a$. In line 3, if the two values are not equal, the program jumps to an error state \textsf{err}.
	
	\begin{lstlisting}[frame=single, caption=Reading the upper stack, label=asm1]
	mov eax, [sp - 8]
	cmp eax, a
	je err
	\end{lstlisting}
	
	Using a simple PDS model, it is not possible to know what is being read. However, our UPDS model and the previous algorithms provide us with reasonable approximations which can be used to examine possible values stored in $\textsf{eax}$, as shown in Figure \ref{fig:3-ex2}.
	
	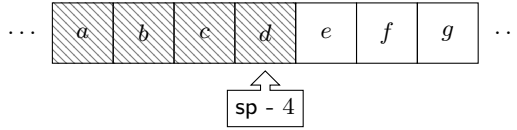
\begin{figure}[!hth]
		\centering
		\begin{minipage}{.8\linewidth}
			\centering
			\begin{tikzpicture}
			\edef\sizetape{0.8cm}
			\tikzstyle{tmtape}=[draw,minimum size=\sizetape]
			\tikzstyle{tmhead}=[arrow box,draw,minimum size=.45cm,arrow box
			arrows={north:0.25cm}]
			
			\begin{scope}[start chain=1 going right,node distance=0mm, outer sep=0mm]
			\node [on chain=1,tmtape,draw=none] {$\ldots$};
			\node [on chain=1,tmtape,pattern=north west lines, pattern color=gray] {$a$};
			\node [on chain=1,tmtape,pattern=north west lines, pattern color=gray]{$b$};
			\node [on chain=1,tmtape,pattern=north west lines, pattern color=gray]{$c$};
			\node [on chain=1,tmtape,pattern=north west lines, pattern color=gray] (input) {$d$};
			\node [on chain=1,tmtape]{$e$};
			\node [on chain=1,tmtape] {$f$};
			\node [on chain=1,tmtape] {$g$};
			\node [on chain=1,tmtape,draw=none] {$\ldots$};
			\end{scope}
			
			\node [tmhead,yshift=-0.60cm] at (input.south){\textsf{sp} - 4};
			\end{tikzpicture}
			\captionof{figure}{The stack being read.}
			\label{fig:3-ex2}
		\end{minipage}
	\end{figure}
	
	To check whether this program reaches the error state $err$ or not, we define the regular set $\mathcal{X} = P \times \Gamma^* a \times \Gamma^*$ of forbidden configurations where $a$ is present on the upper stack just above the stack pointer. If the intersection of the under-approximation of $pre^* ( \mathcal{X} )$ with the set of starting configurations $\mathcal{C}$ of the program is not empty, then \textsf{eax} can contain a critical value, and the program is unsafe. On the other hand, if the intersection of the over-approximation of $post^* ( \mathcal{C} )$ with the set $\mathcal{X}$ is empty, then the program can be considered safe.

	\subsection{Changing the stack pointer}
	
	Another malicious use of the stack pointer \textsf{sp} would be to change the starting point of the stack. As an example, the instruction \textsf{mov sp, sp - 12} changes the stack pointer in such a manner that, from the configuration of Figure \ref{fig:3-ex3a}, the top three elements above it now belong to the stack, as shown in Figure \ref{fig:3-ex3b}.
	
	\begin{figure}[!hth]
		\centering
		\begin{minipage}{.8\linewidth}
			\centering
			\begin{tikzpicture}
			\edef\sizetape{0.8cm}
			\tikzstyle{tmtape}=[draw,minimum size=\sizetape]
			\tikzstyle{tmhead}=[arrow box,draw,minimum size=.45cm,arrow box
			arrows={north:0.25cm}]
			
			\begin{scope}[start chain=1 going right,node distance=0mm, outer sep=0mm]
			\node [on chain=1,tmtape,draw=none] {$\ldots$};
			\node [on chain=1,tmtape,pattern=north west lines, pattern color=gray] {$a$};
			\node [on chain=1,tmtape,pattern=north west lines, pattern color=gray] {$b$};
			\node [on chain=1,tmtape,pattern=north west lines, pattern color=gray] {$c$};
			\node [on chain=1,tmtape,pattern=north west lines, pattern color=gray] {$d$};
			\node [on chain=1,tmtape] (input) {$e$};
			\node [on chain=1,tmtape] {$f$};
			\node [on chain=1,tmtape] {$g$};
			\node [on chain=1,tmtape,draw=none] {$\ldots$};
			\end{scope}
			
			\node [tmhead,yshift=-.55cm] at (input.south){\textsf{sp}};
			\end{tikzpicture}
			\captionof{figure}{The original stack.}
			\label{fig:3-ex3a}
		\end{minipage}
	\end{figure}
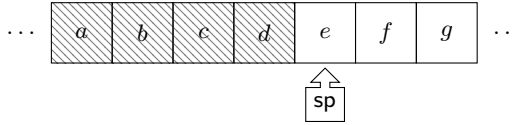
	
	\begin{figure}[!hth]
		\centering
		\begin{minipage}{.8\linewidth}
			\centering
			\begin{tikzpicture}
			\edef\sizetape{0.8cm}
			\tikzstyle{tmtape}=[draw,minimum size=\sizetape]
			\tikzstyle{tmhead}=[arrow box,draw,minimum size=.45cm,arrow box
			arrows={north:0.25cm}]
			
			\begin{scope}[start chain=1 going right,node distance=0mm, outer sep=0mm]
			\node [on chain=1,tmtape,draw=none] {$\ldots$};
			\node [on chain=1,tmtape,pattern=north west lines, pattern color=gray] {$a$};
			\node [on chain=1,tmtape] (input) {$b$};
			\node [on chain=1,tmtape] {$c$};
			\node [on chain=1,tmtape] {$d$};
			\node [on chain=1,tmtape] {$e$};
			\node [on chain=1,tmtape] {$f$};
			\node [on chain=1,tmtape] {$g$};
			\node [on chain=1,tmtape,draw=none] {$\ldots$};
			\end{scope}
			
			\node [tmhead,yshift=-0.55cm] at (input.south){\textsf{sp}};
			\end{tikzpicture}
			\captionof{figure}{After changing \textsf{sp}.}
			\label{fig:3-ex3b}
		\end{minipage}
	\end{figure}
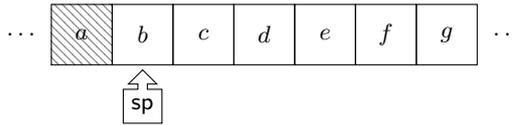
	
	If we model a program as a UPDS, then using our previous algorithms to compute approximations of the reachability set would allow us to have an approximation of the content of the new stack after the stack pointer change.

	\section{Related work}
	
	One way to improve the expressiveness of pushdown automata is to change the way transition rules interact with the stack. Ginsburg et al. introduced in \cite{GGH-1967} \emph{stack automata} that can read the inside of their own stack using a moving stack pointer but can only modify the top. As shown in \cite{hopcroft-1968}, stack automata are equivalent to linear bounded automata (LBA). A LBA is a non-deterministic Turing machine whose tape is bounded between two end markers that cannot be overwritten. This model cannot simulate a UPDS whose lower stack is of unbounded height.
	
	Uezato et al. defined in \cite{UM-atva13} \emph{pushdown systems with transductions}: in such a model, a finite transducer is applied to the whole stack after each transition. However, this model is Turing powerful unless the transducers used have a finite closure, in which case it is equivalent to a simple pushdown system. When the set of transducers has a finite closure, this class cannot be used to simulate UPDSs.
	
	\emph{Multi-stack automata} have two or more stacks that can be read and modified, but are unfortunately Turing powerful. Following the work of Qadeer et al. in \cite{QR-tacas05}, La Torre et al. introduced in \cite{TMP-lics07} \emph{multi-stack pushdown systems with bounded phases}: in each phase of a run, there is at most one stack that is popped from. Anil Seth later proved in \cite{seth-cav10} that the $pre^*$ of a regular set of configurations of a multi-pushdown system with bounded phases is regular; we use this result to perform a bounded phase analysis of our model.
	
	\emph{2-visibly pushdown automata} (2-VPDA) were defined by Carotenuto et all. in \cite{CMP-dlt07} as a variant of two-stack automata where the stack operations are driven by the input word. However, an ordering constraint on the stacks that prevent a 2-VPDA from simulating a UPDS has to be added in order to solve the emptiness problem or the model-checking problem.

	\section{Conclusion}
	
	The first contribution of this paper is a more precise pushdown model of the stack of a program as defined in Section 1. We then investigate the sets of predecessors and successors of a regular set of configurations of an UPDS. Unfortunately, we prove that neither of them are regular. However, we show that the set of successors is context-sensitive. As a consequence, we can decide whether a single configuration is forward reachable or not in an UPDS.
	
	We then prove that the set of predecessors of an UPDS is regular given a limit of $k$ phases, where a phase is a part of a run during which either pop or push rules are forbidden. Bounded-phase reachability is an under-approximation of the actual reachability relation on UPDSs that we can use to detect some incorrect behaviours.
	
	We also give an algorithm to compute an over-approximation of the set of successors. Our idea is to first over-approximate the runs of the UPDS, then compute an over-approximation of the reachable upper stack configurations from this abstraction of runs and consider its product with the regular, accurate and computable set of lower stack configurations.
	
	Finally, we use these approximations on programs to detect stack overflow errors as well as malicious attacks that rely on stack pointer manipulations.

	\bibliography{references_corrected}

\end{document}